\documentclass[11pt, a4paper]{article}

\pdfoutput=1

\usepackage[margin=1in]{geometry}
\usepackage[utf8]{inputenc}
\usepackage[T1]{fontenc}
\usepackage[colorlinks]{hyperref}
\usepackage{microtype}
\usepackage{xcolor}

\hypersetup{
  linkcolor=[rgb]{0.3,0.3,0.6},
  citecolor=[rgb]{0.2, 0.6, 0.2},
  urlcolor=[rgb]{0.6, 0.2, 0.2}
}

\usepackage{amsthm}

\usepackage[T1]{fontenc}
\usepackage{lmodern}
\usepackage{mathtools, amsthm, amsfonts, amssymb}

\usepackage{mathrsfs}
\usepackage{pbox}

\usepackage{authblk}
\usepackage{booktabs}
\usepackage{braket}
\usepackage{commath}
\usepackage{caption}
\usepackage{titlesec}
\titlelabel{\thetitle. }
\titleformat*{\section}{\large\bfseries}

\usepackage[capitalise]{cleveref}
\usepackage{accents}

\usepackage{tikz}
\usetikzlibrary{positioning}

\usepackage{chngcntr}
\usepackage{apptools}
\AtAppendix{\counterwithin{theorem}{section}}

\newtheorem{theorem}{Theorem}[section]
\newtheorem*{theorem*}{Theorem}
\newtheorem{proposition}[theorem]{Proposition}
\newtheorem{lemma}[theorem]{Lemma}
\newtheorem{corollary}[theorem]{Corollary}

\newtheorem{definition}[theorem]{Definition}

\newtheorem*{problem*}{Problem}

\newtheorem{remark}[theorem]{Remark}

\DeclarePairedDelimiter\floor{\lfloor}{\rfloor}

\DeclareMathAccent{\wtilde}{\mathord}{largesymbols}{"65}

\DeclareMathOperator{\supp}{\text{supp}}

\newcommand{\C}{\mathbb{C}}

\newcommand{\N}{\mathbb{N}}

\newcommand{\GHZ}{\mathrm{GHZ}}

\newcommand{\R}{\mathbb{R}}

\newcommand{\cB}{\mathcal{B}}
\newcommand{\cH}{\mathcal{H}}

\newcommand{\cX}{\mathcal{X}}
\newcommand{\cY}{\mathcal{Y}}

\newcommand{\cS}{\mathcal{S}}
\newcommand{\lv}{\lvert}
\newcommand{\rv}{\rvert}
\newcommand{\lV}{\lVert}
\newcommand{\rV}{\rVert}
\newcommand{\ra}{\rangle}
\newcommand{\la}{\langle}

\DeclareMathOperator{\Hom}{Hom}
\DeclareMathOperator{\Diag}{Diag}

\DeclareMathOperator{\Tr}{Tr}
\DeclareMathOperator{\spank}{span}

\newcommand{\loccto}{\xrightarrow{\textnormal{LOCC}}}

\newcommand{\ketbra}[2]{\left|#1\middle\rangle\!\middle\langle#2\right|}

\newcommand{\bipartite}[1]{\psi_{#1}}

\usepackage{letltxmacro}
\LetLtxMacro\orgvdots\vdots
\LetLtxMacro\orgddots\ddots

\makeatletter
\DeclareRobustCommand\vdots{%
  \mathpalette\@vdots{}%
}
\newcommand*{\@vdots}[2]{%
  \sbox0{$#1\cdotp\cdotp\cdotp\m@th$}%
  \sbox2{$#1.\m@th$}%
  \vbox{%
    \dimen@=\wd0 %
    \advance\dimen@ -3\ht2 %
    \kern.5\dimen@
    \dimen@=\wd2 %
    \advance\dimen@ -\ht2 %
    \dimen2=\wd0 %
    \advance\dimen2 -\dimen@
    \vbox to \dimen2{%
      \offinterlineskip
      \copy2 \vfill\copy2 \vfill\copy2 %
    }%
  }%
}
\DeclareRobustCommand\ddots{%
  \mathinner{%
    \mathpalette\@ddots{}%
    \mkern\thinmuskip
  }%
}
\newcommand*{\@ddots}[2]{%
  \sbox0{$#1\cdotp\cdotp\cdotp\m@th$}%
  \sbox2{$#1.\m@th$}%
  \vbox{%
    \dimen@=\wd0 %
    \advance\dimen@ -3\ht2 %
    \kern.5\dimen@
    \dimen@=\wd2 %
    \advance\dimen@ -\ht2 %
    \dimen2=\wd0 %
    \advance\dimen2 -\dimen@
    \vbox to \dimen2{%
      \offinterlineskip
      \hbox{$#1\mathpunct{.}\m@th$}%
      \vfill
      \hbox{$#1\mathpunct{\kern\wd2}\mathpunct{.}\m@th$}%
      \vfill
      \hbox{$#1\mathpunct{\kern\wd2}\mathpunct{\kern\wd2}\mathpunct{.}\m@th$}%
    }%
  }%
}

\author[1]{Asger Kj\ae rulff Jensen}
\author[2,1]{P\'eter Vrana}
\affil[1]{QMATH, Department of Mathematical Sciences, University of Copenhagen, Universitetsparken 5, 2100 Copenhagen, Denmark} \affil[2]{Department of Geometry, Budapest University of Technology and Economics, Egry J\'ozsef u. 1., 1111 Budapest, Hungary}

\title{The asymptotic spectrum of LOCC transformations}

\makeatother
\begin{document}
\maketitle

\begin{abstract}
We study exact, non-deterministic conversion of multipartite pure quantum states into one-another via local operations and classical communication (LOCC) and asymptotic entanglement transformation under such channels. In particular, we consider the maximal number of copies of any given target state that can be extracted exactly from many copies of any given initial state as a function of the exponential decay in success probability, known as the converese error exponent. We give a formula for the optimal rate presented as an infimum over the asymptotic spectrum of LOCC conversion. A full understanding of exact asymptotic extraction rates between pure states in the converse regime thus depends on a full understanding of this spectrum. We present a characterisation of spectral points and use it to describe the spectrum in the bipartite case. This leads to a full description of the spectrum and thus an explicit formula for the asymptotic extraction rate between pure bipartite states, given a converse error exponent. This extends the result on entanglement concentration in \cite{MR1960075}, where the target state is fixed as the Bell state. In the limit of vanishing converse error exponent the rate formula provides an upper bound on the exact asymptotic extraction rate between two states, when the probability of success goes to 1. In the bipartite case we prove that this bound holds with equality.
\end{abstract}
\section{Introduction}
The primary objects of study in this paper are entangled $k$-partite pure states of finite dimensional quantum systems, represented by vectors in the tensor product of $k$ finite dimensional Hilbert spaces. Allowing for certain quantum operations of these systems yields a resource theory, viewing quantum states as resources and quantum operations as methods of extracting one resource from another. One important set of quantum operations are LOCC channels, which allow for local application of completely positive maps and the sending of classical information between parties. For any two resources $\ket{\psi}$ and $\ket{\phi}$, one can ask the question: Can $\ket{\psi}$ be transformed into $\ket{\phi}$? And asymptotically: How many copies of $\ket{\phi}$ can be extracted per copy of $\ket{\psi}$? Some parts of these questions have been answered both, mostly in the bipartite case, $k=2$, \cite{MR1960075} \cite{PhysRevA.53.2046} \cite{PhysRevLett.83.436} \cite{PhysRevLett.83.1046}. Inherent in the application of quantum measurements is the uncertainty of outcomes. For this reason some resources can only be converted to others with certain probability, and the question of asymptotic extraction therefore depends on how one demands that the probability of successful conversion behaves asymptotically. Given a probability distribution $P=(p_i)_{i=1}^d$, we consider the bipartite pure state $\ket{\bipartite{P}}=\sum_i \sqrt{p_i}\ket{ii}$. In \cite{MR1960075}, the following formula was derived for the number of EPR pairs that can be asymptotically extracted from $\ket{\bipartite{P}}$ given that the success probability behaves as $2^{-nr+o(n)}$:
\begin{equation}
E^*(r) = \inf_{\alpha\in [0,1)}\frac{r\alpha +\log\sum_i p_i^\alpha}{1-\alpha}.
\end{equation}
$r$ is known as the converse error exponent.
\\

For two multipartite pure states, $\ket{\psi}$ and $\ket{\phi}$, we let $E^*(r,\psi,\phi)$ be the number of copies of $\ket{\phi}$ that can be asymptotically extracted per copy of $\ket{\psi}$ with success probability behaving like $2^{-nr+o(n)}$. 
In Theorem \ref{mainTheorem1} we show that for $k$ parties, the optimal rate can be expressed as
\begin{equation}\label{equation15}
E^*(r,\psi,\phi) =
\inf_{f\in \Delta(\cS_k)} \frac{r\alpha(f)  + \log f(\ket{\psi})}{\log f(\ket{\phi})},
\end{equation}
where $\Delta(\cS_k)$ is a certain set of functions, which we call the asymptotic LOCC spectrum. In 
particular, an explicit description of the points of $\Delta(\cS_k)$ would 
imply a complete understanding of the asymptotic extraction rates. In Theorem \ref{UniversalPointsTheorem} we present a characterization of the functions in $\Delta(\cS_k)$.
\\

In the bipartite case the characterization from Theorem \ref{UniversalPointsTheorem} allows for an explicit description of $\Delta(\cS_2)$ (Theorem \ref{bipartiteSpectrum}). We present the following formula for the extraction rate between $\ket{\bipartite{P}}$ and $\ket{\bipartite{Q}}$ with converse error exponent $r$, generalising the formula for $E^*(r)$:
\begin{equation}
E^*(r,\bipartite{P},\bipartite{Q}) = \inf_{\alpha\in [0,1)}\frac{r \alpha +\log\sum_i p_i^\alpha}{\log\sum_i q_i^\alpha}.
\end{equation}

While it does not follow from the general theory, one might reasonably conjecture that in the $r\to 0$ limit, one obtains the optimal extraction rate with success probability going to 1. In the bipartite case we show that this is indeed true. Expressed with R\'enyi entropies in Theorem \ref{extractionRate}, the formula for the optimal extraction rate between bipartite pure states with success probability going to 1 is
\begin{equation}
E(\bipartite{P},\bipartite{Q}) = \min_{\alpha\in [0,1]}\frac{H_\alpha(P)}{H_\alpha(Q)}.
\end{equation}
This result much resembles the formula conjectured in \cite[Example 8.26]{fritz_2017} and proven in \cite{2018arXiv180805157K}, where the minimum is taken over all of $[0,\infty]$, describing the extraction rate under the condition, that the probability of success is identically $1$ for sufficiently many copies.
\\

Our results are inspired by the work of Strassen on the asymptotic 
restriction problem for tensors \cite{strassen1988asymptotic}. In that paper he establishes a 
characterization in terms of the asymptotic spectrum associated with the 
semiring of equivalence classes of tensors, equipped with the preorder 
induced by tensor restriction. We prove that the semiring $\cS_k$ of local 
unitary equivalence classes of pure unnormalized states, equipped with 
the preorder induced by LOCC convertibility satisfies similar 
properties, which leads to the characterization in (\ref{equation15}). 

\section{The asymptotic LOCC spectrum}
In \cite{strassen1988asymptotic}, Strassen considers the semiring of equivalence classes of tensors under invertible local linear transformations. This is a semiring with respect to direct sum and tensor products and the preorder, given by convertibility via local linear transformations, respects the algebraic structure of this semiring. By applying the spectral theorem \cite[Theorem 2.3]{strassen1988asymptotic}, one gets the asymptotic spectrum $\Delta(B)$ of tensors. First we recall the necessary definitions and state the theorem.
\begin{definition}
	A commutative semiring $(\cS,+,\cdot)$ is a set $\cS$ with two binary, commutative, and associative operations $(+,\cdot)$ containing distinct additive and multiplicative identity elements $0,1\in\cS$, satisfying the distributive law:
	\begin{equation}
		a(b+c)=a b+a c.
	\end{equation}
\end{definition}
	Note that what distinguishes a semiring from a ring, is that there is no guarantee of an additive inverse. In fact the semiring we will consider in this paper has no additive inverses, except for 0.
	In this paper all semirings are commutitative and semiring shall therefore be understood to implicitly mean commutative semiring.
\begin{definition}
	A preorder $\le$ on $\cS$ is a binary relation which is transistive and reflexive (but not necessarily antisymmetric).
	We say that $(\cS,+,\cdot,\le)$ is a preordered semiring, if $\le$ respects the algebraic structure on $\cS$. That is, when $a\le b$ and $c\le d$:
	\begin{align}
	a+c&\le b+d \label{OrderedSemiringsum}
	\\ac&\le bd.\label{OrderedSemiringproduct}
	\end{align}
\end{definition}
\begin{remark}\label{remarkSemiring}
	Note that in order to show conditions (\ref{OrderedSemiringsum}) and (\ref{OrderedSemiringproduct}) it suffices to show $a+c\le b+c$ and $ac\le bc$ whenever $a\le b$, since this implies $a+c\le b+c\le b+d$ whenever $a\le b$ and $c\le d$, and similarly for the product.
\end{remark}
One can always turn a semiring into a preordered semiring by defining $\le$ to be either the equality preorder ($x\le y\iff x=y$) or the other extreme preorder ($\forall x,y\in \cS: x\le y$). We shall only be interested in certain non-trivial preorders, namely semirings where $\N\subset \cS$ and the preorder restricted to $\N$ is the usual ordering of $\N$.
\begin{theorem}[{Strassen, \cite{strassen1988asymptotic}, see also \cite[Theorem 2.2]{zuiddam2018graph}}]\label{Stone}
	Let $(\cS,\le)$ be a preordered semiring with $\N\subset \cS$ satisfying the following:
	\begin{enumerate}
		\item $\le$ restricted to $\N$ is the usual ordering of $\N$. \label{item1}
		\item For any $a,b\in \cS\textbackslash \{0\}$ there is an $r\in \N$ such that $a\le rb$.\label{item2}
	\end{enumerate}
	Define the asymptotic preorder $\lesssim$ on $\cS$ by; $a\lesssim b$ if and only if $a^N\le 2^{x_N}b^N$ for some interger-valued sequence $x_N\in o(N)$. Then $(\cS,\lesssim)$ is also a preordered semiring. Let
	\begin{align*}
	\Delta(\cS)&=\left\{f \in\Hom(\cS,\R^+)\lv\forall a,b\in \cS: a\le b \implies f(a)\le f(b)\right\}.
	\end{align*}
	Then
	\begin{equation}\label{equation6}
		a\lesssim b\quad \iff\quad \forall f\in\Delta(\cS): f(a)\le f(b).
	\end{equation}
	Let $\Delta(\cS)$ be equipped with the topology generated by the maps $\hat a:\Delta(B)\to \R$, given by $\hat a:f\mapsto f(a)$. That is, $\Delta(\cS)$ is equipped with the coarsest topology making these maps continuous. Then $\Delta(\cS)$ is a compact Hausdorff space and $a\mapsto \hat a$ is a semiring homomorphism $\cS\to C(\Delta(\cS))$, which, by (\ref{equation6}), respects both $\lesssim$ and $\le$ on $\cS$. $\Delta(\cS)$ will be called the asymptotic spectrum of $\cS$.
\end{theorem}
\begin{remark}
	In this paper, we are interested in the asymptotic ordering of $\cS$ and (\ref{equation6}) is therefore the important property of the asymptotic spectrum. The topology on $\Delta(\cS)$ will not play a role, but the fact that $\cS$ maps into $C(\Delta(\cS))$ in an order preserving manner explains the use of the term spectrum.
\end{remark}
The goal of this section is to study the semiring of local unitary orbits of unnormalised pure states with preorder defined by LOCC convertibility. We show that this is a preordered semiring and that the conditions of Theorem \ref{Stone} are satisfied, yielding an LOCC spectrum.
\begin{definition}
	We define states to be positive elements
	\[
	\rho\in \cB(\cH_1)\otimes\cdots\otimes\cB(\cH_k)\otimes \Diag(\C^\cX),
	\]
	where $(\cH_i)_{i=1}^k$ are finite dimensional Hilbert spaces, $\cX$ is a finite set and  $\Diag(\C^\cX)=\spank\{\ketbra{x}{x} \ \lv\ x\in\cX\}$ is the space of diagonal matrices acting on $\C^\cX$. The cone of positive elements in\\
	$\cB(\cH_1)\otimes\cdots \otimes\cB(\cH_k)\otimes \Diag(\C^\cX)$ will be called a k-partite state space.
\end{definition}
The $\cH_i$'s are to be viewed as physically seperated quantum systems and $\cX$ as a classical register. Notice that we are not demanding that the states are normalized to $\Tr(\rho)=1$. We associate to each unnormalized state, $\rho$, the normalized state $\frac{\rho}{\Tr{\rho}}$. As we consider LOCC conversions between unnomalized states, the ratio of the traces will correspond to the probability of successful conversion between the normalized states.
\begin{definition}\label{LOCCDefinition}
	We define a one-step LOCC channel to be a map \[\Lambda:\cB(\cH_1)\otimes\cdots \otimes\cB(\cH_k)\otimes \Diag(\C^\cX)\to\cB(\cH_1)\otimes\cdots\otimes\cB(\cH'_i)\otimes\cdots\otimes\cB(\cH_k)\otimes \Diag(\C^{\cY})\] between two state spaces that is given by
	\[
	\Lambda:\rho\mapsto \sum_{j\in J} \left((K_j)_i\otimes \ketbra{g(j)}{f(j)}\right)\rho\left( ({K_j}^*)_i\otimes \ketbra{f(j)}{g(j)}\right).
	\] 
	Here $i\le k$ is a positive integer, $J$ is a finite index set and $f:J\to \cX$, $g:J\to\cY $ are maps. For each $j\in J$, $K_j:\cH_i\to \cH'_i$ is a linear map and $(K_j)_i=I_{\cH_1}\otimes\cdots \otimes K_j\otimes \cdots  \otimes I_{\cH_k}$. These maps need to satisfy
	\begin{equation}\label{KrausConstraint}
	\sum_{j\in J} {K_j}^*K_j\otimes \ketbra{f(j)}{f(j)} \le I,
	\end{equation}
	where $I$ denotes the identity operator on $\cH_i\otimes \C^\cX$. The operators $K_j$ are called the Kraus operators.
	\\\\
	An LOCC protocol is a finite sequence of composable one-step LOCC channels and the composition is an LOCC channel.
\end{definition}
The reader may have noticed that the definition we use differs from that found elsewhere in the literature. In our formulation the local channels do not depend explicitly on earlier rounds, but instead act jointly on the local system and the shared classical register, which is supposed to store the required measurement results. Conveniently, this also removes the need to model the passing of classical messages between the parties. 
One can think of a one-step LOCC channel in our sense as the act of reading the classical register ($f$), applying the channel given by the Kraus operators $K_j$ (with $j$ in the preimage of the classical variable under $f$) and writing a function of the measurement result and the old value into the register ($g$). Allowing a sequence of such transformations is clearly equivalent to the usual notion of LOCC.
If $g$ is injective, this corresponds to remembering the outcome of each measurement. There is little to gain from forgetting the measurement outcomes, and therefore it is often convenient to use one-step LOCC channels of the following form.
\begin{definition}
We say that a one-step LOCC channel, $\Lambda$, is \textit{remembering} if the Kraus operators are indexed over $J=\cY$ and $g$ is the identity map on $\cY$. That is
\begin{equation}
	\Lambda:\rho\mapsto \sum_{y\in \cY} \left((K_y)_i\otimes \ketbra{y}{f(y)}\right)\rho\left( ({K_y}^*)_i\otimes \ketbra{f(y)}{y}\right).
\end{equation}
\end{definition}
Given two states $\rho_1$ and $\rho_2$, we say that $\rho_2$ can be extracted from $\rho_1$ under LOCC and write $\rho_1\loccto \rho_2$ if there exists an LOCC channel $\Lambda$, such that $\Lambda(\rho_1)=\rho_2$. Under the identification $\cB(\cH_1)\otimes\cdots \otimes\cB(\cH_k)\simeq\cB(\cH_1)\otimes\cdots \otimes\cB(\cH_k)\otimes \Diag(\C)$ we shall also consider positive elements of the former as states.
\\

To any vector $\ket{\psi}\in \cH_1\otimes\cdots \otimes\cH_k$ we associate the pure state $\ketbra{\psi}{\psi} \in \cB(\cH_1)\otimes\cdots \otimes\cB(\cH_k)$ and we write $\ket{\psi} \loccto \ket{\phi}$ if the corresponding statement is true for their respective states.
\begin{remark}
	Note that we allow for trace non-increasing completely positive maps. So $\ket{\psi} \loccto \ket{\phi}$ means that we can convert $\frac{\ket{\psi}}{\lv\lv \psi\lv\lv}$ to $\frac{\ket{\phi}}{\lv\lv \phi\lv\lv}$ with success probability $\frac{\lv\lv \phi\lv\lv^2}{\lv\lv \psi\lv\lv^2}$.
\end{remark}
\begin{definition}
	Given $k\in \N$ and finite dimensional Hilbert spaces $\cH_1,\cH'_1,\ldots ,\cH_k,\cH'_k$ we say that $\ket{\phi}\in\cH_1\otimes\cdots \otimes\cH_k$ and $\ket{\psi}\in \cH'_1\otimes\cdots \otimes\cH'_k$ are locally unitarily equivalent, if there exist partial isometries $U_j:\cH_j\to \cH'_j$ such that
	\[
	\ket{\psi} = (U_1\otimes\cdots \otimes U_k)\ket{\phi}
	\]
	and
	\[
	\ket{\phi} = (U_1^*\otimes\cdots \otimes U_k^*)\ket{\psi}.
	\]
	Let $\cS_k$ denote the set of equivalence classes.
\end{definition} 
Note that for any two representatives, $[\ket{\psi}]=[\ket{\phi}]$, of an element of $\cS_k$, the partial isometries witnessing this equivalence define $k$-step LOCC channels mapping one to the other and back; $\ket{\psi}\loccto \ket{\phi}\loccto\ket{\psi}$. In other words, states that are locally unitarily equivalent are also LOCC-equivalent. The following preorder is therefore well-defined:
\[
	\left[\ket{\psi}\right]\ge\left[\ket{\phi}\right] \text{	iff		} \ket{\psi} \loccto\ket{\phi}.
\]
By \cite[Corollary 1]{PhysRevA.63.012307}, LOCC equivalence also implies local unitary equivalence. So the above preorder is in fact a partial order. This is not of importance for the theory to work, but still worth noting.
\\\\
When $\ket{\psi}\in \cH_1\otimes\cdots \otimes\cH_k$ and $\ket{\phi}\in \cH'_1\otimes\cdots \otimes \cH'_k$ we may take direct sum and tensor product to get new $k$-partite states
\[
\ket{\psi}\oplus\ket{\phi}\in (\cH_1\oplus\cH'_1)\otimes\cdots \otimes(\cH_k\oplus\cH'_k)
\]
\[
\ket{\psi}\otimes\ket{\phi}\in (\cH_1\otimes\cH'_1)\otimes\cdots \otimes(\cH_k\otimes\cH'_k).
\]
Both sum and product respect local unitary equivalence, turning $(\cS_k,\oplus,\otimes)$ into a semiring. We wish to apply Theorem \ref{Stone} to $(\cS_k,\oplus,\otimes,\le)$. For this purpose, what remains to be shown is that $(\cS_k,\oplus,\otimes,\le)$ is a preordered semiring and that conditions \ref{item1} and \ref{item2} of Theorem \ref{Stone} are satisfied. We start out by showing that it is a preordered semiring. (\ref{OrderedSemiringproduct}) is immediate, so we proceed to proving (\ref{OrderedSemiringsum}), which is done in Proposition \ref{DirectSumWithFixed}.
\\\\
We say that a state $\rho\in \cB(\cH_1)\otimes\cdots \otimes\cB(\cH_k)\otimes \text{Diag}(\C^\cX)$ is conditionally pure if it can be written in the form
\begin{equation}
	\rho = \sum_{x\in \cX} \ketbra{\phi_x}{\phi_x}\otimes\ketbra{x}{x}.
\end{equation}
Proposition \ref{SimpleLOCCprotocols} is well known and establishes that we can restrict our attention to protocols which keep track of all measurements until finally throwing away the register. For completeness we provide a proof of Proposition \ref{SimpleLOCCprotocols}. The proof will be induction on the following lemma:
\begin{lemma}\label{RemeberingLemma1}
	Let $\Lambda=\Lambda_2\circ\Lambda_1$ be a two-step LOCC channel. Then $\Lambda = \Lambda_2'\circ\Lambda_1'$, for some two-step LOCC protocol $(\Lambda_2',\Lambda_1')$ where $\Lambda_1'$ is remembering.
\end{lemma}
\begin{proof}
	Let $A_r=\cB(\cH_1^r)\otimes\cdots\otimes \cB(\cH_k^r)\otimes \Diag(\C^{\cX_r})$ for $r=0,1,2$ be the three state spaces in question. 
	\begin{equation}
		A_0\stackrel{\Lambda_1}{\to} A_1\stackrel{\Lambda_2}{\to} A_2. 
	\end{equation}
	Let $\Lambda_r$ be defined as in Definition \ref{LOCCDefinition} by $f_r,g_r,J_r,i_r$ and $(K_j^r)_{j\in J_r}$ for $r=1,2$. We first expand the register, $\cX_1$ of $A_1$: Let 
	\begin{equation}
		J=\cX=\{(j_1,j_2)\in J_1\times J_2| f_2(j_2)=g_1(j_1)\}
	\end{equation}
	and let $A_1'=\cB(\cH_1^1)\otimes\cdots\otimes \cB(\cH_k^1)\otimes \Diag(\C^{\cX})$. Define the remembering channel $\Lambda'_1:A_0\to A_1'$ by the index map $f'_1:\cX\to \cX_0$ given as $f'_1:(j_1,j_2)\mapsto f_1(j_1)$. 
	The Kraus operators for $\Lambda_1'$ are $(K^1_x)_{x\in \cX}$, 
	where $K^1_{(j_1,j_2)}=K^1_{j_1}$. Define $\Lambda_2'$ in a similar manner: The Kraus operators $(K_x^2)_{x\in\cX}$ are indexed over $\cX$ with $K_{(j_1,j_2)}^2=K_{j_2}^2$ and applied via the index maps $f_2'=\text{id}:\cX\to\cX$ and $g_2':(j_1,j_2)\mapsto g_2(j_2)$. Now
	\begin{equation}
	\begin{split}
		&\Lambda_2'\circ\Lambda_1'(\rho)
		\\
		=&\sum_{(j_1,j_2)\in\cX} \bigg[(K_{j_2}^2)_{i_2}(K_{j_1}^1)_{i_1}\otimes\ketbra{g_2(j_2)}{f_1(j_1)}\bigg]
		\rho
		\bigg[(K_{j_1}^1)_{i_1}^*({K_{j_2}^2})_{i_2}^*\otimes\ketbra{f_1(j_1)}{g_2(j_2)}\bigg]
		\\
		=&\Lambda_2\circ\Lambda_1(\rho).
	\end{split}
	\end{equation}
\end{proof}
\begin{remark}\label{TraceRemark}
	Note that for the construction in the proof of Lemma \ref{RemeberingLemma1}, if $\cX_1=\cX_2$ is a one-point set and $A_1=A_2$ and $\Lambda_2=I_{A_1}$ is the identity channel, then $\Lambda_2'=\Tr_{\Diag(\C^\cX)}$ is just the partial trace of the register.
\end{remark}
\begin{proposition}\label{SimpleLOCCprotocols}
	Given a channel $\Lambda$ for which the final state space has a one-point register, there exists an LOCC protocol $(\Lambda_n,\ldots,\Lambda_1)$, consisting of remembering one-step LOCC channels, such that
	\begin{equation}
		\Lambda = \Tr_{\Diag(\C^{\cX_n})}\circ \Lambda_n\circ\cdots\circ\Lambda_1.
	\end{equation}
	Here $\cX_i$ is the $i$'th register and $\Tr_{\Diag(\C^{\cX_n})}$ is the partial trace of the final register.
\end{proposition}
\begin{proof}
	Let $\Lambda$ be the composition of an $n$-step LOCC protocol and let $A_n$ be the final state space. Then $\Lambda=I_{A_n}\circ\Lambda$. By applying Lemma \ref{RemeberingLemma1} $n$ times and by Remark \ref{TraceRemark}
	\begin{equation}
		\Lambda = \Tr_{\Diag(\C^{\cX_n})}\circ\Lambda_n\circ\cdots\circ\Lambda_1,
	\end{equation}
	where $\Lambda_i$ is a remembering one-step channel for $i=1,\ldots, n$. 
\end{proof}
\begin{proposition}\label{DirectSumWithFixed}
	Let $\ket{\phi_1}$, $\ket{\phi_2}$ and $\ket{\psi}$ be k-tensors, then
	\[\ket{\phi_1}\loccto\ket{\phi_2}\quad \implies\quad\ket{\phi_1}\oplus\ket{\psi}\loccto\ket{ \phi_2}\oplus\ket{\psi}.\]
\end{proposition}
\begin{proof}
	By Proposition \ref{SimpleLOCCprotocols},
	\begin{equation}
		\ketbra{\phi_2}{\phi_2} = \Tr_{\Diag(\C^{\cX_n})}\circ \Lambda_n\circ\cdots\circ\Lambda_1 \ketbra{\phi_1}{\phi_1}
	\end{equation}
	for some remembering protocol $(\Lambda_n,\ldots,\Lambda_1)$. This implies that
	\begin{equation}
		\Lambda_n\circ\cdots\circ\Lambda_1 \ketbra{\phi_1}{\phi_1} = \sum_{y\in \cX_n } a_y\ketbra{\phi_2}{\phi_2}\otimes \ketbra{y}{y}
	\end{equation}
	for some $a_y\ge 0$ with $\sum_{y}a_y=1$. It suffices to show that there exists some LOCC channel $\Lambda'$ such that 
	\begin{equation}
		\Lambda' \ketbra{\phi_1\oplus\psi}{\phi_1\oplus \psi} = \sum_{y\in \cX_n } a_y\ketbra{\phi_2\oplus \psi}{\phi_2\oplus \psi}\otimes \ketbra{y}{y}.
	\end{equation}
	This is shown by induction on $n$. For $n=0$ it is trivial. Assume that it is possible for $(n-1)$-step protocols. Let $(K_x)_{x\in \cX_1}$ be the Kraus operators for $\Lambda_1$ acting on system $i$.
	\begin{equation}
		\Lambda_1\ketbra{\phi_1}{\phi_1} = \sum_{x\in \cX_1} \ketbra{\phi_x}{\phi_x}\otimes \ketbra{x}{x}
	\end{equation}
	where $\ket{\phi_x}=(K_x)_{i}\ket{\phi_1}$ and $\sum \la \phi_x|\phi_x\ra\le \la \phi_1|\phi_1\ra$.
	For each $x\in \cX_1$
	\begin{equation}
		\Lambda_n\circ\cdots\circ\Lambda_2 \ketbra{\phi_x}{\phi_x} =\sum_{y\in \cY_x} a_y\ketbra{\phi_2}{\phi_2}\otimes \ketbra{y}{y},
	\end{equation}
	where $(\cY_x)_{x\in\cX_1}$ is a partition of $\cX_n$. Let $c_x=\sum_{y\in \cY_x} a_y$. By the induction hypothesis there exist channels $(\Lambda_x')_{x\in\cX_1}$ such that
	\begin{equation}
	\Lambda_x' \Big[\ketbra{\phi_x\oplus \sqrt{c_x}\psi}{\phi_x\oplus \sqrt{c_x} \psi}\otimes\ketbra{x}{x}\Big] = \sum_{y\in \cY_x } a_y\ketbra{\phi_2\oplus \psi}{\phi_2\oplus \psi}\otimes \ketbra{y}{y}.
	\end{equation}
	Define $\Lambda_1'$ by the Kraus operators $K_x'=K_x\oplus \sqrt{c_x} I$ for each $x\in \cX_1$, where $I$ is the identity operator acting on the $i$'th system on which $\ket{\psi}$ lives. Then
	\begin{equation}
		\Lambda_1'\ketbra{\phi_1\oplus \psi}{\phi_1\oplus \psi}=\sum_{x\in \cX_1}\ketbra{\phi_x\oplus \sqrt{c_x}\psi}{\phi_x\oplus \sqrt{c_x}\psi}\otimes \ketbra{x}{x}.
	\end{equation}
	Let $\tilde{\Lambda}$ be the LOCC channel
	\begin{equation}
		\tilde{\Lambda}:\sum_{x\in \cX_1}\rho_x\otimes \ketbra{x}{x}\mapsto \sum_{x\in \cX_1} \Lambda_x'\Big[\rho_x\otimes \ketbra{x}{x}\Big]
	\end{equation}
	Now
	\begin{equation}
	\begin{split}
		\tilde{\Lambda}\circ\Lambda_1'\ketbra{\phi_1\oplus\psi}{\phi_1\oplus\psi} 
		&
		=
		\tilde{\Lambda}\sum_{x\in \cX_1}\ketbra{\phi_x\oplus \sqrt{c_x}\psi}{\phi_x\oplus \sqrt{c_x}\psi}\otimes \ketbra{x}{x}
		\\&
		=
		\sum_{x\in \cX_1}\sum_{y\in \cY_x} a_y\ketbra{\phi_2\oplus \psi}{\phi_2\oplus \psi}\otimes \ketbra{y}{y}
		\\&
		=\sum_{y\in \cX_n} a_y\ketbra{\phi_2\oplus \psi}{\phi_2\oplus \psi}\otimes \ketbra{y}{y}.
	\end{split}
	\end{equation}
\end{proof}
By Remark \ref{remarkSemiring} it follows that equation (\ref{OrderedSemiringsum}) holds for $(\cS_k,\oplus,\otimes,\le)$, which is therefore a preordered semiring.
\\

It remains to be shown that conditions \ref{item1} and \ref{item2} in Theorem \ref{Stone} are satisfied. The multiplicative unit in $\cS_k$ is represented by the pure state $\ket{0\ldots 0}\in \C^{\otimes k}$ and the additive unit is represented by the zero-vector $0\in \C^{\otimes k}$. $\N$ embeds into $\cS_k$ in the following sense: An integer $d\in \N$ is represented in $\cS_k$ by the $d$-level, $k$-partite, unnormalized GHZ state
\begin{equation}
	\ket{\GHZ_d} =\sum_{i=0}^{d-1} \ket{i\ldots i}\in (\C^d)^{\otimes k},
\end{equation}
and $\ket{\GHZ_{d_1}}\loccto \ket{\GHZ_{d_2}}$ iff $d_1\ge d_2$, so \ref{item1} holds.
\\

We proceed by proving that condition \ref{item2} holds:
\begin{proposition}
	For any non-zero pure states $\ket{\psi}$ and $\ket{\phi}$, there is a $d\in \N$ such that
	\begin{equation}
		\ket{\GHZ_d}\otimes \ket{\psi} \loccto \ket{\phi}.
	\end{equation}
\end{proposition}
\begin{proof}
	By having one party locally construct the normalized $\ket{\phi}$, converting $\operatorname{GHZ}$ states to $\operatorname{EPR}$ pairs between parties and using quantum teleportation \cite{PhysRevLett.70.1895} \cite[s. 6.5.3]{Nielsen:2011:QCQ:1972505}, one obtains a protocol that extracts the normalized version of $\ket{\phi}$. Furthermore $\ket{\psi}\loccto \lV\psi\rV\ket{\GHZ_1}$. So for sufficiently large $d$
	\begin{equation}
	\ket{\GHZ_d}\otimes \ket{\psi}\loccto\frac{1}{\lV \phi\rV}\ket{\phi}\otimes\ket{\psi}\loccto \frac{\lV\psi\rV}{\lV \phi\rV}\ket{\phi}.
	\end{equation}
	In order to obtain $\ket{\phi}$ one simply increases $d$ to $dn$ for large enough $n$ and traces out the $\GHZ$ states not used for teleportation:
	\begin{equation}
	\ket{\GHZ_{dn}}\otimes \ket{\psi}=\ket{\GHZ_{n}}\otimes\ket{\GHZ_{d}}\otimes \ket{\psi}\loccto \frac{\lV\psi\rV}{\lV \phi\rV}\ket{\GHZ_{n}}\otimes\ket{\phi}\loccto 2^{n/2}\frac{\lV\psi\rV}{\lV \phi\rV}\ket{\phi}.
	\end{equation}
	And for $n>2\log\frac{\lV \phi\rV}{\lV \psi\rV}$
	\begin{equation}
		2^{n/2}\frac{\lVert\psi\rVert}{\lVert \phi\rVert}\ket{\phi}\loccto \ket{\phi}.
	\end{equation}
\end{proof}
Theorem \ref{Stone} now applies to $\cS_k$.
\begin{theorem}\label{mainTheorem1}
	Let $\Delta(S_k)$ be the set of order preserving semiring homomorphisms $\cS_k\to \R^+$. Then
	\begin{equation}
		\left[\ket{\psi}\right]\gtrsim \left[\ket{\phi}\right] \iff \forall f\in \Delta(\cS_k): f\left(\ket{\psi}\right)\ge f\left(\ket{\phi}\right).
	\end{equation}
We call $\Delta(\cS_k)$ the asymptotic LOCC spectrum.
\end{theorem}
Concretely $[\ket{\psi}]\gtrsim[\ket{\phi}]$ means that
\begin{equation}
\ket{\operatorname{GHZ}_2}^{\otimes o(n)}\otimes \ket{\psi}^{\otimes n}\loccto\ket{\phi}^{\otimes n},
\end{equation}
where $\ket{\GHZ_2}=\ket{0\ldots 0}+\ket{1\ldots 1}$ is the unnormalized two-level GHZ state. In other words; to extract $n$ copies of $\ket{\psi}$, we need $n$ copies of $\ket{\phi}$, a proportionally vanishing number of GHZ states and the success probability decays as $2^{n(\log\lV \psi\rV^2-\log\lV \phi\rV^2)+o(n)}$.
Since we only need a proportionally vanishing amount of GHZ states we may, assuming that $\ket{\psi}$ is globally entangled, i.e. that tracing out any number of subsystems always leaves a mixed state, extract these $\GHZ$ states from $\ket{\phi}^{\otimes n}$ without further cost in the asymptotic limit. Indeed, one can show that when $\ket{\psi}$ is globally entangled, $\ket{\psi}^{\otimes k}\loccto x\ket{\GHZ_2}$ for some $x>0$. 
\\\\
That is, for any globally entangled $\ket{\psi}$
\[
E^*(r,\psi,\phi) = \sup\left\{\tau\in \R^+\big\lv 2^{nr/2+o(n)}\ket{\psi}^{\otimes n}\loccto\ket{\phi}^{\otimes \floor{\tau n}} \text{ for sufficiently large } n \right\}.
\] 
This implies
\begin{equation}\label{equation4}
E^*(r,\psi,\phi) = \sup\left\{\tau\in \R^+\big\lv \forall f\in \Delta(\cS_k): f(2^{r/2}\ket{\psi})\ge f(\ket{\phi})^\tau \right\}.
\end{equation}
The formula also holds for states that are not globally entangled, but the argument gets somewhat lengthy.
\section{Spectral points}\label{UniversalPointsSection}
The goal of this section is to prove Theorem \ref{UniversalPointsTheorem}, which establishes a condition for a semiring homomorphism $f:\cS_k\to \R^+$ to be monotone (i.e. order preserving) and hence define a point in $\Delta(\cS_k)$. 
\begin{theorem}\label{UniversalPointsTheorem}
	Let $f:\cS_k\to \R^+$ be a semiring homomorphism. Then f is monotone if and only if there is an $\alpha\in [0,1]$ such that $f(\sqrt{p}\ket{0\ldots 0}) = p^\alpha$ and
	\begin{equation}\label{SpectrumInequality}
	f(\ket{\phi}) \ge \left(f\big((\Pi)_i\ket{\phi}\big)^{1/\alpha}+f\big((I-\Pi)_i\ket{\phi}\big)^{1/\alpha} \right)^\alpha
	\end{equation}
	for any $\ket{\phi}\in \cH_1\otimes\cdots\otimes\cH_k$, $i\in\{1,\ldots,k\}$ and orthogonal projection $\Pi\in \cB(\cH_i)$.
\end{theorem}
\begin{proposition}\label{IntroduceAlpha}
	Let $f:\cS_k\to \R^+$ be a monotone semiring homomorphism. There is an $\alpha\ge0$ such that
	\begin{equation}
	f\left(\sqrt{p}\ket{\phi}\right)=p^\alpha f\left(\ket{\phi}\right)
	\end{equation}
	for each $\ket{\phi}$ and each $p>0$.
\end{proposition}
\begin{proof}
	Since $p\mapsto f(\sqrt{p}\ket{0\ldots 0})$ is multiplicative, nondecreasing, sends $0$ to $0$ and $1$ to $1$, it follows from the solution to the Cauchy functional equation that
	\[
	f\left(\sqrt{p}\ket{0\ldots 0}\right) = p^\alpha
	\]
	for all $p>0$ and some $\alpha\ge0$. Therefore
	\begin{align*}
	f\left(\sqrt{p}\ket{\phi}\right)&=f\left(\sqrt{p}\ket{\phi}\otimes \ket{0\ldots 0}\right) = f\left(\ket{\phi}\right)f\left(\sqrt{p}\ket{ 0\ldots 0}\right) =p^\alpha f\left(\ket{\phi}\right) f\left(\ket{0\ldots 0}\right)
	\\
	&=p^\alpha f\left(\ket{\phi}\right).
	\end{align*}
\end{proof}
For the proof of Theorem \ref{UniversalPointsTheorem} we introduce the following extension of a monotone homomorphism $f:\cS_k\to \R^+$ to conditionally pure states. Given $f$ such that $f(\sqrt{p}\ket{0\ldots 0}) = p^\alpha$ for some $\alpha>0$, we define
\[
f\left(\sum_{x\in \cX}\ketbra{\phi_x}{\phi_x} \otimes \ketbra{x}{x} \right) = \left(\sum_{x\in \cX} f(\ket{\phi_x})^{1/\alpha}\right)^\alpha
\]
and if $\alpha=0$, we define
\[
f\left(\sum_{x\in \cX}\ketbra{\phi_x}{\phi_x} \otimes \ketbra{x}{x} \right) = \max_{x\in \cX} f\left(\ket{\phi_x}\right).
\]
\begin{proposition}
	The extension of $f$ is multiplicative under tensor product.
\end{proposition}
\begin{proof}
	For $\alpha>0$
	\begin{equation}
	\begin{split}
	&f\left(\left(\sum_{x\in \cX}\ketbra{\phi_x}{\phi_x}\otimes\ketbra{x}{x} \right)
	\otimes 
	\left(\sum_{y\in \cY}\ketbra{\psi_x}{\psi_x}\otimes\ketbra{x}{x} \right)
	\right)
	\\
	&=
	f\left(\sum_{\stackrel{x\in \cX }{y\in \cY}}\left(\ketbra{\phi_x}{\phi_x}\otimes\ketbra{\psi_y}{\psi_y}\right)\otimes\ketbra{xy}{xy} \right)
	\\
	&=
	\left(
	\sum_{\stackrel{x\in \cX }{y\in \cY}}
	f(\ket{\phi_x}\otimes\ket{\psi_y})^{1/\alpha}
	\right)^{\alpha}
	=
	\left(
	\sum_{\stackrel{x\in \cX }{y\in \cY}}
	f(\ket{\phi_x})^{1/\alpha}f(\ket{\psi_y})^{1/\alpha}
	\right)^{\alpha}
	\\
	&=
	\left(
	\sum_{x\in \cX}
	f(\ket{\phi_x})^{1/\alpha}\sum_{y\in \cY}f(\ket{\psi_y})^{1/\alpha}
	\right)^{\alpha}
	=
	\left(
	\sum_{x\in \cX}
	f(\ket{\phi_x})^{1/\alpha}
	\right)^\alpha\left(\sum_{y\in J}f(\ket{\psi_y})^{1/\alpha}
	\right)^{\alpha}
	\\
	&=
	f\left(\sum_{x\in \cX}\ketbra{\phi_x}{\phi_x}\otimes\ketbra{x}{x} \right) 
	f\left(\sum_{y\in \cY}\ketbra{\phi_y}{\phi_y}\otimes\ketbra{y}{y} \right).
	\end{split}
	\end{equation}
	If $\alpha=0$, then
	\begin{equation}
	\begin{split}
		&
		f\left(\left(\sum_{x\in \cX}\ketbra{\phi_x}{\phi_x}\otimes\ketbra{x}{x} \right)
		\otimes 
		\left(\sum_{y\in \cY}\ketbra{\phi_y}{\phi_y}\otimes\ketbra{y}{y} \right)
		\right)
		\\
		=&
		\max_{\stackrel{x\in \cX}{y\in \cY}}
		f\left(\ket{\phi_x}\right)f\left(\ket{\psi_y}\right)
		=
		\max_{x\in \cX}f\left(\ket{\phi_x}\right)\max_{y\in \cY}f\left(\ket{\psi_y}\right)
		\\
		=&
		f\left(\sum_{x\in \cX}\ketbra{\phi_x}{\phi_x}\otimes\ketbra{x}{x} \right) 
		f\left(\sum_{y\in \cY}\ketbra{\phi_y}{\phi_y}\otimes\ketbra{y}{y} \right).
	\end{split}
	\end{equation}
\end{proof}
The following theorem is true for general LOCC channels, but since it is not needed in full generality, we shall prove it only for remembering LOCC channels.
\begin{proposition}\label{monotoneExtension}
	If $f$ is monotone, then the extension is monotone under remembering LOCC channels.
\end{proposition}
\begin{proof}
	First assume $\alpha>0$ and start with the case where the initial state is pure:
	\begin{equation}
		\ketbra{\psi}{\psi}\loccto \sum_{i\in I}P(i)\ketbra{\phi_i}{\phi_i}\otimes \ketbra{i}{i}.
	\end{equation}
	Here the $\ket{\phi_i}$'s are normalized and $P:I\to \R^+$ is a map. 
	\\
	Given $n\in \N$ we say that a probability measure $Q:I\to \R^+$ is an $n$-type, if $nQ(i)\in \N$ for each $i\in \N$. Given an $n$-type $Q$, we say that a sequence in $I^{n}$ is of type $Q$, if $i$ appears $nQ(i)$ times. The type class $T^n_Q\subset I^n$ is the set of sequences of type $Q$. Given any $n$-type $Q$:
	\begin{equation}\label{equation7}
	\begin{split}
		\ketbra{\psi}{\psi}^{\otimes n} &\loccto 
		\left(\sum_{i\in I}P(i)\ketbra{\phi_i}{\phi_i}\otimes \ketbra{i}{i}\right)^{\otimes n}
		\\
		&=
		\sum_{a\in I^n}\prod_{j=1}^n P(a_j)\bigotimes_{j=1}^n \ketbra{\phi_{a_j}}{\phi_{a_j}}\otimes \ketbra{a}{a}
		\\
		&\loccto
		\lv T^n_Q\lv 2^{-n(H(Q)+D(Q\lv\lv P))}\bigotimes_{i\in I}\ketbra{\phi_i}{\phi_i}^{\otimes nQ(i)}.
	\end{split}
	\end{equation}
	The last LOCC transformation is the projection onto the multiindices of type $Q$ followed by a unitary reshuffling of indicies and a partial trace on the classical register. $H(Q)=-\sum_i Q(i)\log Q(i)$ is the Shannon entropy of $Q$ and $D(Q\lv\lv P)=\sum Q(i)\log\frac{Q(i)}{P(i)}$ is the relative entropy. Since the last expression in (\ref{equation7}) is a pure state we can apply monotonicity of $f$ on pure states to get
	\begin{equation}
		f(\ket{\psi})^{n}\ge \bigg(\lv T^n_Q\lv 2^{-n(H(Q)+D(Q||P))}\bigg)^\alpha\prod_{i\in I}f(\ket{\phi_i})^{nQ(i)}.
	\end{equation}
	Since $\lv T_Q^n\rv \ge 2^{nH(Q)-|I|\log(n+1)}$ \cite[Lemma 4]{MR1960075}, this implies, by taking the $n$-th root;
	\begin{equation}\label{equation1}
		f(\ket{\psi})\ge
		\Big(2^{-D(Q\lv \lv P)+\sum_i Q(i)\log f(\ket{\phi_i})^{1/\alpha}}\Big)^\alpha 2^{-\alpha|I|\frac{\log(n+1)}{n}}.
	\end{equation}
	Let $Z=\sum_{i\in I} P(i) f(\ket{\phi_i})^{1/\alpha}$ and let $P_\phi$ be the probability distribution $P_\phi(i)=\frac{P(i) f(\ket{\phi_i})^{1/\alpha}}{Z}$. Then
	\begin{equation}\label{equation2}
		-D\left(Q\lVert P\right) + \sum_i Q(i)\log f\left(\ket{\phi_i}\right)^{1/\alpha} 
		= -D\left(Q\lV ZP_\phi\right) 
		= \log Z - D\left(Q\lV P_\phi\right).
	\end{equation}
	Using (\ref{equation2}), (\ref{equation1}) becomes
	\begin{equation}\label{equation3}
		f\left(\ket{\psi}\right) \ge 2^{\left(\log Z - D(Q\lV P_\phi)\right)\alpha}2^{-\alpha|I|\frac{\log(n+1)}{n}}.
	\end{equation}
	For each $n\in \N$, let $Q_n$ be an $n$-type with $\supp Q_n=\supp P_\phi$ such that $\lim_n Q_n = P_\phi$. Then $D(Q_n\lV P_\phi)\to D(P_\phi\lV P_\phi)=0$. Inserting $Q_n$ in (\ref{equation3}) and letting $n\to \infty$ yields
	\begin{equation}
		f(\ket{\psi})\ge Z^{\alpha} = \bigg[\sum_{i\in I} P(i) f(\ket{\phi_i})^{1/\alpha}\bigg]^{\alpha} = f\bigg( \sum_{i\in I}P(i)\ketbra{\phi_i}{\phi_i}\otimes \ketbra{i}{i}\bigg),
	\end{equation}
	showing that the extension is monotone under remembering one-step LOCC channels applied to pure states.
	We use this result to generalize to remembering one-step LOCC channels on conditionally pure states:
	\begin{equation}
		\sum_{j\in J}\ketbra{\psi_j}{\psi_j}\otimes\ketbra{j}{j} \loccto
		\sum_{j\in J}\sum_{i\in I_j}\ketbra{\phi_{i,j}}{\phi_{i,j}} \otimes \ketbra{ij}{ij}.
	\end{equation}
	By restricting the protocol to only the Kraus operators acting on $\ket{\psi_j}$ one gets
	\begin{equation}
		\ketbra{\psi_j}{\psi_j}\loccto\sum_{i\in I_j}\ketbra{\phi_{i,j}}{\phi_{i,j}}\otimes\ketbra{i}{i}.	
	\end{equation}
	Therefore
	\begin{align*}
		f\bigg( \sum_{j\in J}\ketbra{\psi_j}{\psi_j} \otimes\ketbra{j}{j} \bigg)
		&=
		\bigg(\sum_{j\in J} f(\ket{\psi_j})^{1/\alpha} \bigg)^\alpha
		\\
		&\ge
		\bigg(\sum_{j\in J}\sum_{i\in I_j} f(\ket{\phi_{i,j}})^{1/\alpha} \bigg)^\alpha
		\\
		&=
		f\bigg( \sum_{j\in J}\sum_{i\in I_j}\ketbra{\phi_{i,j}}{\phi_{i,j}}\otimes\ketbra{ij}{ij} \bigg).
	\end{align*}
	For the case $\alpha=0$, note that 
	\begin{equation}
	\ketbra{\psi_j}{\psi_j}\loccto\sum_{i\in I_j}\ketbra{\phi_{i,j}}{\phi_{i,j}}\otimes\ketbra{i}{i}	
	\end{equation}
	implies $\ket{\psi_j}\loccto \ket{\phi_{i,j}}$ for each $i$, which by monotonicity of $f$ on pure states implies $f(\ket{\psi_j})\ge \max_i f(\ket{\phi_{i,j}})$. Therefore
	\begin{align*}
	f\bigg( \sum_{j\in J}\ketbra{\psi_j}{\psi_j} \otimes\ketbra{j}{j} \bigg)
	&=
	\max_{j} f(\ket{\psi_j})
	\\
	&\ge
	\max_{i,j} f(\ket{\phi_{i,j}})
	\\
	&=
	f\bigg( \sum_{j\in J}\sum_{i\in I_j}\ketbra{\phi_{i,j}}{\phi_{i,j}}\otimes\ketbra{ij}{ij} \bigg).
	\end{align*}
\end{proof}
\begin{lemma}\label{projectionLemma}
	Let $f:\cS_k\to \R^+$ be a semiring homomorphism with $f(\sqrt{p}\ket{0\ldots 0})=p^\alpha$ for some $\alpha\in[0,1]$ which satisfies (\ref{SpectrumInequality}) for any choice of pure state and orthogonal projection. Then
	\begin{equation}\label{SpectrumInequality2}
	f\big(\ket{\phi}\big) \ge \bigg(f\big((A)_i\ket{\phi}\big)^{1/\alpha}+f\big((B)_i\ket{\phi}\big)^{1/\alpha} \bigg)^\alpha
	\end{equation}
	for any $\ket{\phi}\in \cH_1\otimes\cdots\otimes\cH_k$, $i\in\{1,\ldots,k\}$ and $A,B\in \cB(\cH_i)$ with $A^*A+B^*B\le I$.
\end{lemma}
\begin{proof}
	Consider the operator $U=\begin{bmatrix}
	A\\
	B\\
	\sqrt{I-A^*A-B^*B}
	\end{bmatrix}:\cH_i\to\cH_i^3$.
	This is an isometry, so $f(\ket{\phi})=f(\ket{\psi})$, where $\ket{\psi}=(U)_i\ket{\phi}$. Let $\Pi:\cH_i^3\to \cH_i^3$ be the projection onto the first summand. Then $[(\Pi)_i\ket{\psi}]= [(A)_i\ket{\phi}]$ and $[(I-\Pi)_i\ket{\psi}]\ge [(B)_i\ket{\phi}]$, so
	\begin{equation}
	\begin{split}
		f\big(\ket{\phi}\big)=f(\ket{\psi})&\ge\bigg(f\big((\Pi)_i\ket{\psi}\big)^{1/\alpha}+f\big((I-\Pi)_i\ket{\psi}\big)^{1/\alpha} \bigg)^\alpha
		\\
		&\ge \bigg(f\big((A)_i\ket{\phi}\big)^{1/\alpha}+f\big((B)_i\ket{\phi}\big)^{1/\alpha} \bigg)^\alpha.
	\end{split}
	\end{equation}
\end{proof}
\begin{proof}[Proof of Theorem \ref{UniversalPointsTheorem}]
	Suppose $f$ is monotone, then by Proposition \ref{IntroduceAlpha} there is an $\alpha\ge 0$ such that $f(\sqrt{p}\ket{\phi})=p^\alpha f(\ket{\phi})$ for all $\ket{\phi}$ and $p>0$. Consider the extension of $f$ to conditionally pure states. Let $\ket{\phi}$, $i$, and $\Pi$ be given as in the statement of the theorem, then 
	\begin{equation}
		\ketbra{\phi}{\phi} \loccto (\Pi)_i \ketbra{\phi}{\phi} (\Pi)_i\otimes \ketbra{0}{0}+ (I-\Pi)_i \ketbra{\phi}{\phi} (I-\Pi)_i\otimes \ketbra{1}{1},
	\end{equation}
	so by monotonicity of the extension of $f$ we get
	\[
	f\big(\ket{\phi}\big) \ge \bigg(f\big((\Pi)_i\ket{\phi}\big)^{1/\alpha}+f\big((I-\Pi)_i\ket{\phi}\big)^{1/\alpha} \bigg)^\alpha.
	\]
	When $\ket{\phi} = \ket{0\ldots 0}+\ket{1\ldots 1}$, $\Pi=\begin{bmatrix}
	1&0\\
	0& 0
	\end{bmatrix}$ and $I-\Pi=\begin{bmatrix}
	0&0\\
	0& 1
	\end{bmatrix}$, we get by (\ref{SpectrumInequality})
	\begin{equation}
	2=f\big(\ket{\phi}\big) \ge \bigg(f\big((\Pi)_1\ket{\phi}\big)^{1/\alpha}+f\big((I-\Pi)_1\ket{\phi}\big)^{1/\alpha}\bigg)^{\alpha}
	=\bigg(f\big(\ket{0\ldots 0}\big)^{1/\alpha}+f\big(\ket{1\ldots 1}\big)^{1/\alpha}\bigg)^{\alpha}
	= 2^\alpha,
	\end{equation}
	showing that $\alpha\le 1$. This concludes the proof of the ``only if'' statement.
	\\\\
	Conversely, suppose $f$ is a homomorphism satisfying equation (\ref{SpectrumInequality}). By Lemma \ref{projectionLemma}, $f$ satisfies (\ref{SpectrumInequality2}). Consider the extension of $f$ to conditionally pure states. By Lemma \ref{SimpleLOCCprotocols} we need only check that $f$ is monotone under remembering one-step channels and monotone when tracing out the register of a state of the form $\sum_ia_i\ketbra{\phi}{\phi}\otimes\ketbra{i}{i}$. $f$ is monotone under the latter, since 
	\begin{equation}
		\bigg(\sum_i f\big(\sqrt{a_i}\ket{\phi}\big)^{1/\alpha}\bigg)^\alpha = \Big(\sum a_i\Big)^\alpha f(\ket{\phi}) = f\Big(\sum_i\sqrt{a_i}\ket{\phi}\Big).
	\end{equation} 
	For monotonicity under remembering one-step channels, we first consider pure states, that is, we need to show:
	\begin{equation}\label{InductionHypothesis}
		f\big(\ket{\phi}\big)\ge \bigg( \sum_{i=1}^r f\Big((K_i)_j\ket{\phi}\Big)^{1/\alpha}\bigg)^\alpha
	\end{equation}
	whenever $\sum_i K_i^*K_i\le I$. Assume for the sake of induction that (\ref{InductionHypothesis}) is true for $d-1$ and let $(K_i)_{i=1}^d$ be Kraus operators with $\sum_i K_i^*K_i\le I$. Set
	\begin{equation}
		A=\sqrt{\sum_{i=1}^d K_i^*K_i},
	\end{equation}
	\begin{equation}
		B=K_d,
	\end{equation}
	and
	\begin{equation}
		\tilde{K}_i=K_iA^{-1}\quad i=1,\ldots ,d-1.
	\end{equation}
	Here $A^{-1}$ denotes the Moore--Penrose pseudoinverse. Since
	\begin{equation}
		\sum_{i=1}^{d-1} \tilde K_i^*\tilde K_i = A^{-1}\sum_{i=1}^{d-1}  K_i^* K_i A^{-1} = A^{-1}A^2A^{-1}\le I,
	\end{equation}
	we may apply the induction hypothesis to the operators $(\tilde K_i)_i^{d-1}$ and the vector $A_j\ket{\phi}$ to obtain
	\begin{equation}
	\begin{split}
		f\big(\ket{\phi}\big)
		&\ge \bigg( f\big(A_j\ket{\phi}\big)^{1/\alpha}+f\big(B_j\ket{\phi}\big)^{1/\alpha}\bigg)
		\\
		&\ge \Bigg(
		\bigg(
		\bigg(
		\sum_{i=1}^{d-1}f\big((\tilde KA)_j\ket{\phi}\big)^{1/\alpha}
		\bigg)^\alpha
		\bigg)^{1/\alpha}
		+f\big(B_j\ket{\phi}\big)^{1/\alpha}
		\Bigg)^\alpha
		\\
		&=
		\bigg(
		\sum_{i=1}^{d-1}f\big((K_i)_j\ket{\phi}\big)^{1/\alpha}+ f\big((K_d)_j\ket{\phi}\big)^{1/\alpha}
		\bigg)^{\alpha}
		\\
		&=
		\bigg( \sum_{i=1}^d f\big((K_i)_j\ket{\phi}\big)^{1/\alpha}\bigg)^\alpha,
	\end{split}
	\end{equation}
	finishing the induction step.
	\\\\
	Just like in the proof of Proposition $\ref{monotoneExtension}$, this extends to conditionally pure states.
\end{proof}
Note that for $\alpha=0$ an LOCC spectral point is in fact a point in the asymptotic spectrum of tensors in the sense of \cite{CVZ2017} and \cite{strassen1988asymptotic}. For $\alpha=1$ there is just one spectral point, the norm squared:
\begin{proposition}
	Let $f:B^+\to \R$ be a monotone semiring homomorphism with $f(\sqrt{p}\ket{0\ldots0})=pf(\ket{0\ldots0})$ for $p>0$, then 
	\[
		f(\ket{\phi}) = \la \phi|\phi\ra.
	\]
\end{proposition} 
\begin{proof}
	Given $\ket{\phi}$ of norm $1$ we have $\frac{1}{\sqrt{d}}\ket{\GHZ_d}\loccto \ket{\phi} \loccto \ket{0\ldots 0}$ for sufficiently large $d$. Furthermore
	\begin{equation}
		f\Big(\frac{1}{\sqrt{d}}\ket{\GHZ_d} \Big)=\frac{1}{d}\sum_{i=0}^df(\ket{i\ldots i})
		=
		\frac{1}{d}\sum_{i=0}^df(\ket{0\ldots 0}) = f(\ket{0\ldots 0}),
	\end{equation}
	showing that $f(\ket{\phi})=f(\ket{0\ldots 0})=1$. So $f(\sqrt{p}\ket{\phi})=p$ for $p>0$.
\end{proof}
In light of Proposition \ref{IntroduceAlpha}, we get a concrete formula for the extraction rate with converse error exponent $r$:
\begin{equation}\label{equation5}
\begin{split}
E^*(r,\psi,\phi) =& \sup\Big\{\tau\in \R^+\Big\lv \forall f\in \Delta(\cS_k): f(2^{r/2}\ket{\psi})\ge f(\ket{\phi})^\tau \Big\}
\\
=&
\sup\Big\{\tau\in \R^+\Big\lv \forall f\in \Delta(\cS_k): r\alpha(f)+\log f(\ket{\psi})\ge \tau\log f(\ket{\phi}) \Big\}
\\
=&
\inf_{f\in \Delta(\cS_k)} \frac{r \alpha(f) + \log f(\ket{\psi})}{\log f(\ket{\phi})}.
\end{split}
\end{equation}
Here $\alpha(f)=\log f\big(\sqrt{2}\ket{0\ldots 0}\big)$ is the $\alpha$ from Theorem \ref{UniversalPointsTheorem}.
\section{Example: Bipartite states and $\Delta(\cS_2)$}
When $k=2$, we may, by the Schmidt decomposition, write any element in $\cS_2$ as a finite direct sum of terms of the form $\sqrt{p}\ket{0 0}$. Therefore any monotone semiring homomorphism, $f$, is entirely determined by the value of $\alpha(f)\in [0,1]$: For $\ket{\phi}=\ket{\bipartite{P}}=\sum_i \sqrt{p_i}\ket{ii}$ a monotone semiring homomorphism, $f$, must be given by
\begin{equation}
	f(\ket{\phi})= \sum_i p_i^\alpha = \Tr \left[(\Tr_2\ketbra{\phi}{\phi})^\alpha\right],
\end{equation}
where $\Tr_2$ is the partial trace of the second system.\\
The question to answer is then: For which $\alpha\in[0,1]$ does $f_\alpha:\ket{\phi}\mapsto\Tr \left[(\Tr_2\ketbra{\phi}{\phi})^\alpha \right]$ satisfy equation (\ref{SpectrumInequality}). The answer is all of them.
\begin{theorem}\label{bipartiteSpectrum}
	$\Delta(\cS_2)=\{f_\alpha|\alpha\in [0,1]\}$ where
	\[
	f_\alpha:\ket{\phi} \mapsto \Tr \left[(\Tr_2\ketbra{\phi}{\phi})^\alpha \right].
	\]
\end{theorem}
\begin{proof}
	When $\alpha=0$, $f_\alpha(\ket{\psi})$ is the Schmidt rank, which is monotone.
	Assume instead that $\alpha\in (0,1]$. Let $\ket{\phi}\in \C^d\otimes\C^d$ and $\Pi\in \cB(\C^d)$ be an orthogonal projection. It suffices to verify (\ref{SpectrumInequality}) for projections acting on the first system. Let $X\in\cB(\C^d)$ be such that
	\[
		\ket{\phi} = \sum_{i=1}^d X\ket{i}\otimes \ket{i}.
	\]
	Since the coefficients of $\ket{\phi}$ are the square roots of the eigenvalues of $\Tr_2\ketbra{\phi}{\phi}$ and
	\[
	\Tr_2 \ketbra{\phi}{\phi} = \sum_{i=1}^d X\ketbra{i}{i} X^*=XX^*,
	\]
	(\ref{SpectrumInequality}) is equivalent to
	\[
	[\Tr (XX^*)^\alpha]^{1/\alpha}\ge [\Tr (\Pi XX^*\Pi )^\alpha]^{1/\alpha}+[\Tr ((I-\Pi )XX^*(I-\Pi ))^\alpha]^{1/\alpha}.
	\]
	Since $YY^*$ and $Y^*Y$ always have the same eigenvalues we may formulate it instead as
	\[
	[\Tr (X^*X)^\alpha]^{1/\alpha}\ge [\Tr (X^*\Pi X)^\alpha]^{1/\alpha}+[\Tr (X^*(I-\Pi )X)^\alpha]^{1/\alpha}.
	\]
	For $\alpha=1$ this inequality holds since $X^*X=X^*IX= X^*(\Pi +(I-\Pi ))X=X^*\Pi X+X^*(I-\Pi )X$. For $\alpha\in (0,1)$ it follows from \cite[Proposition 3.7]{BH2011}.
\end{proof}
Note that the topology on $\Delta(\cS_2)$ as described in Theorem \ref{Stone} is  the Euclidean topology on $[0,1]$, such that $\Delta(\cS_2)$ can topologically be identified with the unit interval.
\\

Since $\Delta(\cS_2)$ is known we get the following formula for the asymptotic extraction rate between normalized states given converse error exponent $r$. $H_\alpha(P)=\frac{1}{1-\alpha}\log\sum_i p_i^\alpha$ is the $\alpha$-R\'enyi entropy.
\begin{equation}\label{equation9}
E^*(r,\bipartite{P},\bipartite{Q}) = \inf_{\alpha\in [0,1)} \frac{r\alpha+\log\sum p_i^\alpha}{\log\sum q_i^\alpha} = \inf_{\alpha\in [0,1)} \frac{r\frac{\alpha}{1-\alpha}+H_\alpha(P)}{H_\alpha(Q)} .
\end{equation}
When $\ket{\bipartite{Q}}=\frac{1}{\sqrt{2}}(\ket{00}+\ket{11})$ is the maximally mixed state we retrieve the result \cite[eq. (114)]{MR1960075}:
\begin{equation}\label{equation8}
E^*(r,\bipartite{P},\bipartite{Q}) = \inf_{\alpha\in [0,1)} \frac{r\alpha+\log\sum p_i^\alpha}{1-\alpha} =\inf_{\alpha\in [0,1)}  r\frac{\alpha}{1-\alpha}+H_\alpha(P).
\end{equation}
\section{Bipartite states and success probability going to 1}
So far we have considered optimal extraction rates where the success probability is allowed to go to 0. Setting $r=0$ in equation (\ref{equation9}) gives the optimal extraction rate between the two states, where the success rate is allowed to go to 0, but not exponentially fast. This is a good candidate for the optimal extraction rate, when we demand that the success probability goes to 1. Indeed, as we show in Theorem \ref{extractionRate}, setting $r=0$ in (\ref{equation9}) yields the extraction rate between the two states when demanding that success probability goes to 1. For this purpose we make use of the Nielsen's Theorem on LOCC convertibility between bipartite states. The methods of this section are not related to the asymptotic spectrum.
\begin{theorem}[Nielsen, \cite{PhysRevLett.83.436}]\label{Nielsen}
	Let $P=(p_i)_{i=1}^{d_1}$ and $Q=(q_i)_{i=1}^{d_2}$ be two probability distributions, where $p_i$ and $q_i$ are ordered non-increasingly. Then
	\begin{equation}
		\ket{\bipartite{P}}\loccto \ket{\bipartite{Q}} \iff P\preceq Q.
	\end{equation}
	Here $P\preceq Q$ means that $Q$ majorizes $P$, i.e.
	\begin{equation}\label{major}
		\sum_{i=1}^N p_i \le \sum_{i=1}^N q_i
	\end{equation} 
	for all $N$.
\end{theorem}

\begin{proposition}\label{probabilityToOne}
	Let $P=(p_i)_{i=1}^{d_1}$ and $Q=(q_i)_{i=1}^{d_2}$ be two probability distributions which are ordered non-increasingly, and assume that
	\begin{equation}\label{equation10}
		\inf_{\alpha\in [0,1)} \frac{\log f_\alpha(\ket{\bipartite{P}})}{\log f_\alpha(\ket{\bipartite{Q}})} =\inf_{\alpha\in [0,1)} \frac{\log \sum_i p_i^\alpha}{\log\sum_i q_i^\alpha}=\min_{\alpha\in [0,1]} \frac{H_\alpha(P)}{H_\alpha(Q)}>1.
	\end{equation}
	Then for sufficiently large $n$
	\begin{equation}
		\sqrt{x_n}\ket{\bipartite{P}}^{\otimes n}\loccto \ket{\bipartite{Q}}^{\otimes n}
	\end{equation}
	for some sequence of $x_n\ge1$ with $x_n\to 1$. That is; one can asymptotically transform $n$ copies of $\ket{\psi}$ to $n$ copies of $\ket{\phi}$ with probability of success going to 1 as $n\to \infty$.
\end{proposition}
\begin{proof}
	Without loss of generality we can assume that $P$ and $Q$ are non-uniform.
	\begin{equation}
	\ket{\bipartite{P}}^{\otimes n} = \ket{\bipartite{P^{\otimes n}}} = \sum_{I\in {[d_1]}^n} \sqrt{p_{I}} \ket{II},
	\end{equation}
	where $p_I=\prod_{j=1}^n p_{I_j}$. From (\ref{equation10}) we conclude that $H(P)>H(Q)$. Let $V^*>-H(P)$ be chosen such that \cite[Proposition 3.5]{2018arXiv180805157K} applies. Set $t_n=2^{nV^*}$ and note that
	\begin{equation}
	\ket{\bipartite{P}}^{\otimes n}\loccto  \sum_{I\in {[d_1]}^n} \min (\sqrt{p_{I}} ,\sqrt{t_n}) \ket{II}.
	\end{equation}
	Let $x_n=\Big(\sum_{I\in {[d_1]}^n} \min (p_{I},t_n)\Big)^{-1}$ such that 
	\begin{equation}
	\ket{\eta_n}=\sqrt{x_n}\sum_{I\in {[d_1]}^n} \min (\sqrt{p_{I}} ,\sqrt{t_n}) \ket{II}
	\end{equation}
	is normalized and
	\begin{equation}
	\sqrt{x_n}\ket{\bipartite{P}}^{\otimes n} \loccto \ket{\eta_n}.
	\end{equation}
	The proof is complete when it is shown that $x_n\to 1$ and $\ket{\eta_n}\loccto \ket{\bipartite{Q}}^{\otimes n}$ for large $n$. First note that
	\begin{equation}
		\sum_{I\in {[d_1]}^n} \min (p_{I},t_n) \ge  1- \sum_{\stackrel{I\in {[d_1]}^n}{p_I\ge t_n}} p_I.
	\end{equation}
	By \cite[Proposition 2.6, eq. (26)]{2018arXiv180805157K}
	\begin{equation}
		\lim_{n\to \infty}\frac{1}{n}\log\sum_{\stackrel{I\in {[d_1]}^n}{p_I\ge t_n}} p_I <0,
	\end{equation}
	which implies 
	\begin{equation}
		\lim_{n\to \infty}\sum_{\stackrel{I\in {[d_1]}^n}{p_I\ge t_n}} p_I = 0.
	\end{equation}
	So $x_n\to 1$.
	\\\\
	Let $\tilde P_n=\left(\tilde p_{n,I}\right)_{I\in [d_1]^n}=(x_n\min ({p_{I}} ,{t_n}))_{I\in [d_1]^n}$ be the probability distribution, such that $\ket{\eta_n}=\ket{\bipartite{\tilde P_n}}$. By Theorem \ref{Nielsen}, it must be shown that $\tilde P_n\preceq Q^{\otimes n}$ for large $n$.
	\\\\
	Note that $\tilde P_n\preceq P^{\otimes n}$.
	By \cite[Proposition 3.5]{2018arXiv180805157K}, we have for large $n$ and for all $N$ such that $P^{\otimes n\downarrow }(N)\le 2^{nV^*}$
	\begin{equation}\label{eq6}
		\sum_{i=1}^{N-1} \tilde P_n^\downarrow(i)\le\sum_{i=1}^{N-1} P^{\otimes n\downarrow}(i)\le \sum_{i=1}^{N-1} Q^{\otimes n\downarrow}(i).
	\end{equation}
	Let $N^*$ be the largest number such that $P^{\otimes n\downarrow }(N^*)>2^{nV^*}$. By (\ref{eq6})
	\begin{equation}
		\sum_{i=1}^{N^*} \tilde P_n^\downarrow(i) \le \sum_{i=1}^{N^*} Q^{\otimes n\downarrow}(i),
	\end{equation}
	and since $\tilde P_n^\downarrow(i)$ is constant for $i\in [N^*]$, we have
	\begin{equation}
	\sum_{i=1}^{N} \tilde P_n^\downarrow(i) \le \sum_{i=1}^{N} Q^{\otimes n\downarrow}(i)
	\end{equation}
	for all $N\le N^*$.
\end{proof}
\begin{corollary}\label{importantCorollary}
	Given $n,m\in \N$ with 
	\begin{equation}
		\frac{m}{n} > \inf_{\alpha\in [0,1)} \frac{\log f_\alpha(\ket{\psi})}{\log f_\alpha(\ket{\phi})}.
	\end{equation}
	Then
	\begin{equation}
		\sqrt{x_k}\ket{\psi}^{\otimes nk} \loccto \ket{\phi}^{\otimes mk}
	\end{equation}
	for some sequence $x_k\to 1$.
\end{corollary}
Let $E(\psi,\phi)$ denote the optimal rate at which one can extract $\ket{\phi}$ from $\ket{\psi}$ with chance of success going to $1$ asymptotically.
Combining Corollary \ref{importantCorollary} and equation (\ref{equation9}) with $r=0$, yields respectively a lower and upper bound on $E(\psi,\phi)$, which can be summed up as:
\begin{theorem}\label{extractionRate}
	Given probability distributions $P$ and $Q$
	\begin{equation}
		E(\bipartite{P},\bipartite{Q}) = \min_{\alpha\in [0,1]}\frac{H_\alpha(P)}{H_\alpha(Q)}.
	\end{equation}
\end{theorem}
\paragraph{Acknowledgements.}
We acknowledge financial support from the European Research Council (ERC Grant Agreement no. 337603) and VILLUM FONDEN via the QMATH Centre of Excellence (Grant no. 10059). PV acknowledges support from the Hungarian Research Grant NKFI K124152.
\bibliographystyle{ieeetr}
\bibliography{all}
\end{document}